\documentclass{llncs}

\usepackage{amsmath}
\usepackage{amsfonts}
\usepackage{amssymb}

\usepackage{url}

\usepackage{hyperref}
\hypersetup{
  colorlinks=true,
  linkcolor=blue,
  citecolor=green
}

\usepackage[nameinlink]{cleveref}
\crefname{table}{table}{tables}
\Crefname{table}{Table}{Tables}
\crefname{figure}{figure}{figures}
\Crefname{figure}{Figure}{Figures}
\crefname{section}{section}{sections}
\Crefname{section}{Section}{Sections}
\crefname{claim}{claim}{claims}
\Crefname{claim}{Claim}{Claims}
\Crefname{algorithm}{Algorithm}{Algorithms}

\usepackage[ruled,vlined,linesnumbered]{algorithm2e}

\usepackage[usenames,dvipsnames]{xcolor}

\usepackage{tikz}
\usetikzlibrary{shapes,arrows,positioning, decorations.pathmorphing}

\usepackage{enumerate}

\spnewtheorem*{theorem*}{Theorem}{\bfseries}{\rmfamily}
\spnewtheorem{claim}{Claim}{\itshape}{}

\usepackage{multirow}

\newcommand{\Hhillrlx}{\mathbf{H}^{\mathsf{HILL-rlx}}}

\newcommand{\D}{\textsc{D}}
\newcommand{\cD}{\mathcal{D}}
\newcommand{\F}{\textsc{F}}

\newcommand{\SD}{\mathrm{SD}}
\newcommand{\eqd}{\overset{d}{=}}

\newcommand{\nstrings}{\{0,1\}^n}
\newcommand{\mstrings}{\{0,1\}^\ell}

\newcommand{\bigO}[1]{O\left(#1\right)}

\DeclareMathOperator*{\E}{\mathbb{ E }}

\newcommand{\poly}[1]{\mathrm{poly}\left(#1\right)}

\newcommand{\NM}{\textsf{NegativeMass}}

\title{Simulating Auxiliary Inputs, Revisited}

\author{Maciej Skorski 
\thanks{This work was partly supported by the WELCOME/2010-4/2 grant founded within the framework of the EU Innovative Economy Operational Programme.} }
\institute{
\email{maciej.skorski@mimuw.edu.pl} \\ Cryptology and Data Security Group, University of Warsaw}
\begin{document}

\maketitle

\begin{abstract}
For any pair $(X,Z)$ of correlated random variables we can think of $Z$ as a randomized function of $X$. Provided that $Z$ is short, one can make this function computationally efficient
by allowing it to be only approximately correct. In folklore this problem is known as \emph{simulating auxiliary inputs}.
This idea of simulating auxiliary information turns out to be a powerful tool in computer science, finding applications in complexity theory, cryptography, pseudorandomness and zero-knowledge. In this paper we revisit this problem, achieving the following results:

\begin{enumerate}[(a)]
\item We discuss and compare efficiency of known results, finding the flaw in the best known bound claimed in the TCC'14 paper "How to Fake Auxiliary Inputs".
\item We present a novel boosting algorithm for constructing the simulator. Our technique essentially fixes the flaw. This boosting proof is of independent interest, as it shows how to handle "negative mass" issues when constructing probability measures in descent algorithms.
\item Our bounds are much better than bounds known so far. To make the simulator $(s,\epsilon)$-indistinguishable we need the complexity $O\left(s\cdot 2^{5\ell}\epsilon^{-2}\right)$ in time/circuit size, which is better by a factor $\epsilon^{-2}$ compared to previous bounds.
In particular, with our technique we (finally) get meaningful provable security for the  EUROCRYPT'09 leakage-resilient stream cipher instantiated with a standard 256-bit block cipher, like $\mathsf{AES256}$.
\end{enumerate}

Our boosting technique utilizes a two-step approach. In the first step we shift the current result (as in gradient or sub-gradient descent algorithms) and in the separate step we fix the biggest non-negative mass constraint violation (if applicable). 

\end{abstract}

\smallskip

\section{Introduction}

\subsection{Simulating Correlated Information.}
\subsubsection{Informal Problem Statement}
 Let $(X,Z)\in\mathcal{X}\times\mathcal{Z}$ be a pair of correlated random variables. We can think of $Z$ as a \emph{randomized} function of $Z$. More precisely, consider the randomized function $h:\mathcal{X}\rightarrow \mathcal{Z}$, which for every $x$ outputs $z$ with probability $\Pr[Z=z|X=x]$. By definition it satisfies
\begin{align}\label{eq:simulating_nonefficient}
 (X,h(X)) \overset{d}{=} (X,Z)
\end{align}
however the function $h$ is \emph{inefficient} as we need to hardcode the conditional probability table of $Z|X$. It is natural to ask, if this limitation can be overcome
\begin{quote}
\textbf{Q1}: Can we represent $Z$ as an \emph{efficient} function of $X$?
\end{quote}
Not surprisingly, it turns out that a positive answer may be given only in computational settings.
Note that replacing the equality in \Cref{eq:simulating_nonefficient} by closeness in the total variation distance (allowing the function $h$ to make some mistakes with small probability) is not enough \footnote{Indeed, consider the simplest case $\mathcal{Z} = \{0,1\}$, define $X$ to be uniform over $\mathcal{X}=\{0,1\}^n$, and take $Z=f(X)$ where $f$ is a function which is $0.5$-hard to predict by circuits exponential in $n$, 
Then $(X,h(X))$ and $(X,Z)$ are at least $\frac{1}{4}$-away in total variation}. This discussion leads to the following reformulated question
\begin{quote}
\textbf{Q1'}: Can we \emph{efficiently simulate} $Z$ as a function of $X$?
\end{quote}

\subsubsection{Why it matters?}
Aside from being very foundational, this question is relevant to many areas of computer science. We will not discuss these applications in detail, as they are well explained in \cite{JetchevP14}. Below we only mention where such a generic simulator can be applied, to show that this problem is indeed well-motivated.
\begin{enumerate}[(a)]
\item Complexity Theory. From the simulator one can derive Dense Model Theorem \cite{Reingold2008},
Impagliazzo's hardcore lemma \cite{Impagliazzo95} and a version of Sz´emeredi’s Regularity Lemma \cite{FriezeK99}.
\item Cryptography. The simulator can be applied for settings where $Z$ models short leakage from a secret state $X$. It provides tools for improving and simplifying proofs in leakage-resilient cryptography, in particular for leakage-resilient stream ciphers \cite{JetchevP14}. 
\item Pseudorandomness. Using the simulator one can conclude results called chain rules \cite{GentryWichs2010}, which quantify pseudorandomness in conditioned distributions. They can be also applied to leakage-resilient cryptography.
\item Zero-knowledge. The simulator can be applied to represent the text exchanged in verifier-prover interactions $Z$ from the common input $X$ \cite{Chung2015}. 
\end{enumerate}
Thus, the simulator may be used as a tool to unify, simplify and improve many results.
Having briefly explained the motivation we now turn to answer the posed question, leaving a more detailed discussion of some applications to \Cref{subsec:Applications}.

\subsection{Problem Statement}
The problem of simulating auxiliary inputs in the computational setting can be defined precisely as follows
\begin{quote}
Given a random variables $X\in\{0,1\}^n$ and correlated $Z\in\{0,1\}^{\ell}$, what is the minimal complexity $s_h$ of a (randomized) function $h$ such that the distributions of ${h(X)}$ and $Z$ are $(\epsilon,s)$-indistinguishable given $X$, that is 
$$|\E\D(X,{h(X)})-\E\D(X,Z) | < \epsilon$$ holds for all (deterministic) circuits $\D$ of size $s$?
\end{quote}
The indistinguishability above is understood with respect to deterministic circuits. However it doesn't really matter for distinguishing two distributions, where randomized and deterministic distinguishers are equally powerful\footnote{If two distributions can be distinguished by a randomized circuit, we can fix a specific choice of coins to achieve at least the same advantage}.

It turns out that it is relatively easy\footnote{We briefly sketch the idea of the proof: note first that it is easy to construct a simulator for every single distinguisher. Having realized that, we can use the min-max theorem to switch the quantifiers and get one simulator for all distinguishers.} to construct a simulator $h$ with a polynomial blowup in complexity, that is 
when 
$$s_h =\poly{ s,\epsilon^{-1}, 2^{\ell} }.$$
However, more challenging is to minimize the dependency on $\epsilon^{-1}$. This problem is especially important for cryptography, where security definitions require the advantage $\epsilon$ to be possibly small. Indeed, for meaningful security $\epsilon=2^{-80}$ or at least $\epsilon = 2^{-40}$ it makes a difference whether we lose $\epsilon^{-2}$ or $\epsilon^{-4}$.
We will see later how much inefficient bounds here may affect provable security of stream ciphers.


\subsection{Related Works}

\subsubsection{Original work of Jetchev and Pietrzak (TCC'14)} The authors showed that $Z$ can be ``{approximately}'' computed from $X$ by an ``{efficient}'' function $\mathsf{h}$. 
\begin{theorem}[\cite{JetchevP14}, corrected]\label{thm:Simulator_Pietrzak}
For every distribution $(X,Z)$ on $\{0,1\}^n\times\{0,1\}^{\ell}$ and every $\epsilon$, $s$, there exists a ``simulator'' $h:\{0,1\}^n\rightarrow \{0,1\}^\ell$ such that
\begin{enumerate}[(a)]
\item $(X,{h}(X))$  and $(X,Z)$ are $(\epsilon,s)$-indistinguishable 
\item ${h}$ is of complexity $s_{\mathsf{h}}=\bigO{ s\cdot 2^{4\ell}\epsilon^{-4} } $ 
\end{enumerate}
\end{theorem}
The proof uses the standard min-max theorem. In the statement above we correct two flaws. One is a missing factor of $2^{\ell}$. The second (and more serious) one is the (corrected) factor $\epsilon^{-4}$, claimed incorrectly to be $\epsilon^{-2}$. The flaws are discussed in \Cref{sec:flaws}.
\subsubsection{Vadhan and Zheng (CRYPTO'13)} The authors derived a version of \Cref{thm:Simulator_Pietrzak} but with incomparable bounds
\begin{theorem}[\cite{VadhanZheng2013}]\label{thm:Simulator_Vadhan}
For every distribution $X,Z$ on $\{0,1\}^n\times\{0,1\}^{\ell}$ and every $\epsilon$, $s$, there exists a ``simulator'' $h:\{0,1\}^n\rightarrow \{0,1\}^\ell$ such that
\begin{enumerate}[(a)]
\item $(X,{h}(X))$  and $(X,Z)$ are $(s,\epsilon)$-indistinguishable 
\item ${h}$ is of complexity $s_{\mathsf{h}}=\bigO{ s\cdot 2^{\ell}\epsilon^{-2} + 2^{\ell}\epsilon^{-4}} $ 
\end{enumerate}
\end{theorem}
The proof follows from a general regularity theorem which is based on their uniform min-max theorem. The additive loss of $\bigO{2^\ell\epsilon^{-4}}$ appears as a consequence of a sophisticated weight-updating procedure. This error is quite large and may dominate the main term for many settings (whenever $s \ll \epsilon^{-2}$). 

As we show later, \Cref{thm:Simulator_Vadhan} and \Cref{thm:Simulator_Pietrzak} give in fact comparable security bounds when applied to leakage-resilient stream ciphers (see \Cref{subsec:Applications})

\subsection{Our Results}

We reduce the dependency of the simulator complexity $s_h$ on the advantage $\epsilon$ to only a factor of $\epsilon^{-2}$, from the factor of $\epsilon^{-4}$.
\begin{theorem}[Our Simulator]\label{thm:Simulator_ours}
For every distribution $X,Z$ on $\{0,1\}^n\times\{0,1\}^{\ell}$ and every $\epsilon$, $s$, there exists a ``simulator'' $h:\{0,1\}^n\rightarrow \{0,1\}^\ell$ such that
\begin{enumerate}[(a)]
\item $(X,{h}(X))$  and $(X,Z)$ are $(s,\epsilon)$-indistinguishable 
\item ${h}$ is of complexity $s_{\mathsf{h}}=\bigO{ s\cdot 2^{5\ell}\epsilon^{-2}}$ 
\end{enumerate}
\end{theorem}
Below in \Cref{table:simulator_comparison} we compare our result to previous works.
\begin{table}[!h]
\centering
\renewcommand{\arraystretch}{1.2}
\resizebox{0.99\linewidth}{!}{
\begin{tabular}{|c|c|c|c|l|}
\hline
Author & Technique & Advantage & Size & Cost of simulating \\
\hline
\cite{JetchevP14} (\Cref{thm:Simulator_Pietrzak}) & Min-Max & \multirow{3}{*}{$\epsilon$} & \multirow{3}{*}{$s$} & $s_{\mathsf{h}} = \bigO{ s\cdot 2^{{4\ell}} \mathbf{\textcolor{red}{\epsilon^{-4}} } }$ \\
\cline{1-2}\cline{5-5}
\cite{VadhanZheng2013} (\Cref{thm:Simulator_Vadhan}) & Complicated Boosting &  &  & $s_{\mathsf{h}} = \bigO{ s\cdot 2^{\ell} /\epsilon^2 + \textcolor{red}{\mathbf{ 2^{\ell}\epsilon^{-4}}}}$ \\
\cline{1-2}\cline{5-5}
\textbf{This paper} (\Cref{thm:Simulator_ours}) & Simple Boosting &  &  & $s_{\mathsf{h}} = \bigO{ s\cdot 2^{5\ell}\textcolor{ForestGreen}{\mathbf{\epsilon^{-2}} }}$ \\
\hline
\end{tabular}
}
\caption{The complexity of simulating $\ell$-bit auxiliary information given  required indistinguishability strength, depending on the proof technique. 
}
\label{table:simulator_comparison}
\end{table}

Our result is slightly worse in terms of dependency $\ell$, but outperforms previous results in terms of dependency on $\epsilon^{-1}$. However, the second dependency is more crucial for cryptographic applications. Note that the typical choice is sub-logarithmic leakage, that is $\ell = o\left(\log\epsilon^{-1}\right)$ is asymptotic settings\footnote{This is a direct consequence of the fact that we want $\ell$ fits poly-preserving reductions} (see for example \cite{Chung2015}).
Stated in non-asymptotic settings this assumption translates to $\ell < c\log \epsilon^{-1}$ where $c$ is a small constant (for example $c= \frac{1}{12}$ see \cite{Pietrzak2009}). In these settings, we outperform previous results. 

To illustrate this, suppose we want to achieve security $\epsilon = 2^{-60}$ simulating just one bit from a $256$-bit input. As it follows from \Cref{table:simulator_comparison}, previous bounds are useless as they give the complexity bigger than $2^{256}$ which is the worst complexity of all boolean functions over the chosen domain. In settings like this, only our bound can be applied to conclude meaningful results. For more concrete examples of settings where our bounds are even only meaningful, we refer to \Cref{table:2} in \Cref{subsec:Applications}.

\subsection{Our Techniques}

Our approach utilizes a simple boosting technique: as long as the condition (a) in \Cref{thm:Simulator_ours} fails, we can use the distinguisher to improve the simulator. This makes our algorithm constructive with respect to oracle answers, similarly to other boosting proofs.  In short, if we find $\D$ such that 
\begin{align*}
\E\D(X,Z) - \E\D(X,h(X)) > \epsilon
\end{align*}
then we construct $h'$ according to the equation\footnote{As we already mentioned, we can assume that $\D$ is deterministic without loss of generality. Then all the terms in the equation are well-defined.} 
\begin{align*}
\Pr[h' (x)= z] = \Pr[h(x)=z]+\gamma\cdot \mathsf{Shift}\left(\D(x,z)\right) + \mathsf{Corr}(x,z)
\end{align*}
where
\begin{enumerate}[(a)]
\item The parameter $\gamma$ is a\emph{fixed step} chosen in advance (its optimal value depends on $\epsilon$ and $\ell$ and is calculated in the proof.)
\item $ \mathsf{Shift}\left(\D(x,z)\right) $ is a \emph{shifted} version of $\D$, so that $\sum_{z} \mathsf{Shift}\left(\D(x,z)\right)  = 0$. This restriction correspond to the fact that we want to preserve the constraint $\sum_{z}h(x,z)=1$.
 More precisely, $\mathsf{Shift}\left(\D(x,z)\right) =\D(x,z)-\E_{z'\leftarrow U_{\ell}}\D(x,z)$
\item $\mathsf{Corr}(x,z)$ is a \emph{correction term} used to fix (some of) possibly negative weights.
\end{enumerate}
The procedure is being repeated in a loop, over and over again. The main technical difficulty is to show that it eventually stops after not so many iterations. 

Note that in every such a step the complexity cost of the shifting term is $O\left( 2^{\ell}\cdot\mathrm{size}(\D) \right)$\footnote{By definition, it requires computing the average of $\D(x,\cdot)$ over $2^{\ell}$ elements}. In our solution, the correction term does a search over $z$ looking for the biggest negative mass, and redistributes it over the remaining points. Intuitively, it works because the total negative mass is getting smaller with every step. See \Cref{alg:simulating} for a pseudo-code description of the algorithm and the rest of  \Cref{sec:proof:Simulator_ours} for a proof.

\subsection{Applications}\label{subsec:Applications}

\subsubsection{Better security for the EUROCRYPT'09 stream cipher.}

The first construction of leakage-resilient stream cipher was proposed by Dziembowski and Pietrzak in~\cite{Dziembowski2008}. 
On \Cref{fig:SC_Pietrzak} below we present a simplified version of this cipher~\cite{Pietrzak2009}, based on a weak pseudorandom function (wPRF). 
\begin{figure}[!th]
\centering
\begin{tikzpicture}
\node[] (k0) at (0,2) {$K_0$};
\node[] (x0) at (0,1) {$x_0$};
\node[] (k1) at (0,0) {$K_1$};
\node[draw,circle] (f0) at (1,2) {$F$};
\node[draw,circle] (f1) at (3,0) {$F$};
\node[draw,circle] (f2) at (5,2) {$F$};
\node[draw,circle] (f3) at (7,0) {$F$};
\node[circle, text opacity = 0] (f4) at (8,2) {$F$};
\node[circle, text opacity = 0] (f41) at (8,1) {$F$};
\node[circle, text opacity = 0] (f5) at (9,0) {$F$};
\draw[-latex] (k0) -- (f0);
\draw[-latex] (x0) -- (f0);
\draw[-latex] (k1) -- (f1);
\draw[-latex] (f0) -- (f1) node[midway, above right] {$x_1$};
\draw[-latex] (f1) -- (f3) node[midway, above right] {$K_3$};
\draw[-latex] (f1) -- (f2) node[midway, above left] {$x_2$};
\draw[-latex] (f0) -- (f2) node[midway, above] {$K_2$};
\draw[-latex] (f2) -- (f3) node[midway, above right] {$x_3$};
\draw[dotted] (f2) -- (f4) node[midway, above right] {$K_4$};
\draw[dotted] (f3) -- (f5) node[midway, above right] {$K_5$};
\draw[dotted] (f3) -- (f41) node[above left] {$x_5$};
\node[color=gray, above = 1 of f0] (l0) {$L_0$};
\draw[color=gray, decorate, decoration=snake] (f0) -- (l0);
\node[color=gray, above= 1 of f2] (l2) {$L_2$};
\draw[color=gray, decorate, decoration=snake] (f2) -- (l2);
\node[color=gray, below= 1 of f1] (l1) {$L_1$};
\draw[color=gray, decorate, decoration=snake] (f1) -- (l1);
\node[color=gray, below= 1 of f3] (l3) {$L_3$};
\draw[color=gray, decorate, decoration=snake] (f3) -- (l3);
\end{tikzpicture}
\caption{The EUROCRYPT'09 stream cipher (adaptive leakage). $F$ denotes a weak pseudorandom function. By $K_i$ and $x_i$ we denote, respectively, values of the secret state and keystream bits. Leakages are denotted in gray with $L_i$.}
\label{fig:SC_Pietrzak}
\end{figure}
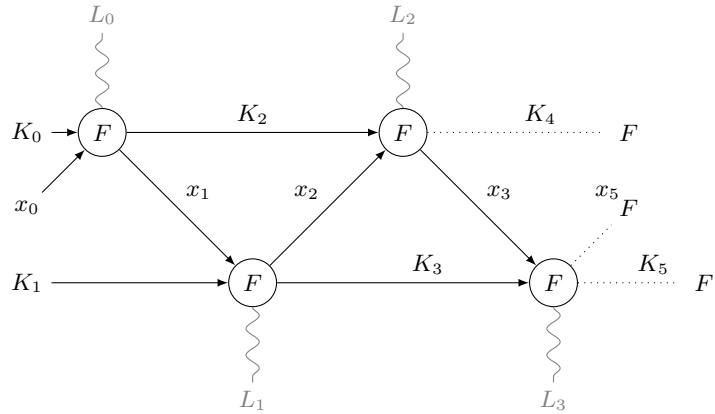

\noindent Jetchev and Pietrzak in \cite{JetchevP14} showed how to use the simulator theorem to simplify the security analysis of the EUROCRYPT'09 cipher. The cipher security depends on the complexity of the simulator as explained in \Cref{thm:Simulator_Pietrzak} and \Cref{remark:exact_loss}. We consider the following setting:
\begin{itemize}
\item number of rounds $q=16$,
\item $F$ instantiated with $\mathsf{AES}256$ (as in \cite{JetchevP14})
\item cipher security we aim for $\epsilon'=2^{-40}$
\item $\lambda = 3$ bits of leakage per round
\end{itemize}
The concrete bounds for $(q,\epsilon',s')$-security of the cipher (which roughly speaking means that $q$ consecutive outputs is $(s',\epsilon')$-pseudorandom, see \Cref{sec:prelim} for a formal definition) are given in \Cref{table:2} below. We ommit calculations as they are merely putting parameters from \Cref{thm:Simulator_Pietrzak}, \Cref{thm:Simulator_Vadhan} and \Cref{thm:Simulator_ours} into \Cref{remark:exact_loss} and assuming that AES as a weak PRF is $(\epsilon,s)$-secure for any pairs $s/\epsilon \approx 2^k$ (following the similar example in \cite{JetchevP14}).
\begin{table}
\renewcommand{\arraystretch}{1.2}
\centering
\begin{tabular}{|l|c|l|l|l|l|l|}
\hline
Analysis/Authors & wPRF security & Leakage & Advantage $\epsilon'$ & Size $s'$ \\ \hline
\cite{JetchevP14} (\Cref{thm:Simulator_Pietrzak}) & \multirow{3}{*}{$256$ } & \multirow{3}{*}{$\lambda=3$ } & \multirow{3}{*}{ $2^{-40}$ } & \textcolor{red}{$0$} \\ \cline{1-1}\cline{5-5}
\cite{VadhanZheng2013} (\Cref{thm:Simulator_Vadhan}) & & &  & \textcolor{red}{$0$} \\ \cline{1-1}\cline{5-5}
\textbf{this paper} (\Cref{thm:Simulator_ours}) & & &  & \textcolor{ForestGreen}{$2^{66}$} \\  \hline
\end{tabular}
\caption{The security of the EUROCRYPT'09 stream cipher, instantiated with AES256 as a weak PRF of rouhgly $k=256$ bits of security. In this settngs only our new bounds provide non-trivial bounds.}\label{table:2}
\end{table}

More generaly, we can give the following comparison of security bounds for different wPRF-based stream ciphers, in terms of time-sccess ratio. The bounds in \Cref{table:1} follow from the simple lemma in
 \Cref{sec:time-success_simpler}, which shows how the time-success ratio changes under explicit reduction formulas.
\begin{table}
\centering
\resizebox{0.95\textwidth}{!}{
\begin{tabular}{|c|l|l|l|c|}
\hline
 Cipher & Analysis & Proof techniques & Security level & Comments  \\
\hline
(1) & \cite{Pietrzak2009} & Pseudoentropy chain rules & $k' \ll \frac{1}{8}k$& large number of blocks \\
\hline
(1) & \cite{JetchevP14} & Aux. Inputs Simulator (corr.) & $k' \approx \frac{k}{6}-\frac{5}{6}\lambda $ & \\
\hline
(1) & \cite{VadhanZheng2013} & Aux. Inputs Simulator & $k' \approx \frac{k}{6}-\frac{1}{3}\lambda $ & \\
\hline
(1) & \textbf{This work} & Aux. Inputs Simulator & $k' \approx \frac{k}{4}-\frac{4}{3}\lambda $ & \\
\hline
(2) & \cite{Faust2012} & Pseudoentropy chain rules & $k' \approx \frac{k}{5} - \frac{3}{5}\lambda$ & large public seed \\
\hline
(3) & \cite{Yu2013} & Square-friendly apps. &  $k' \approx \frac{k}{4}-\frac{3}{4}\lambda$ & only in $\mathsf{minicrypt}$ \\ \hline
\end{tabular}
}
\caption{Different bounds for wPRF-based leakage-resilient stream ciphers. $k$ is the security level of the underlying wPRF. The value $k'$ is the security level for the cipher, understood in terms of time-success ratio. the numbers denote: (1) The EUROCRYPT'09 cipher, (2) The CSS'10/CHESS'12 cipher, (3) The CT-RSA'13 cipher.}
\label{table:1}
\end{table}

\subsection{Organization}

In \Cref{sec:prelim} we discuss basic notions and definitions. The proof of \Cref{thm:Simulator_ours} appears in \Cref{sec:proof:Simulator_ours}. 







\section{Preliminaries}\label{sec:prelim}

\subsection{Basic Notions} 
Let $\mathcal{V}$ be a finite set, and $\cD$ be a class of deterministic real functions on $\mathcal{V}$.
For any two real functions $f_1,f_2$ on $\mathcal{V}$, we say that $f_1,f_2$ are $(\cD,\epsilon)$-indistinguishable if 
\begin{align*}
\forall \D \in \cD:\quad  \left| \E_{x\sim V} \D(x)\cdot f_1(x)-\E_{x\sim V} \D(x)\cdot f_2(x) ) \right| \leqslant \epsilon
\end{align*}
If $\mathcal{D}$ consists of all circuits of size $s$ we say that $f_1,f_2$ are $(s,\epsilon)$-indistinguishable. 

\subsection{Stream ciphers definitions}
We start with the definition of weak pseudorandom functions, which are \emph{computationally indistinguishable} from random functions, when queried on random inputs and fed with uniform secret key.
\begin{definition}[Weak pseudorandom functions]
A function $\F: \{0, 1\}^{k} \times \{0, 1\}^{n} \rightarrow \{0, 1\}^{m}$ is an $(\epsilon, s, q)$-secure weak PRF if its outputs on $q$ random inputs are indistinguishable from random by any distinguisher of size $s$, that is 
\begin{align*}
\left| \Pr \left[\D\left(\left( X_i \right)_{i=1}^{q},\F((K,X_i)_{i=1}^{q} \right)=1\right] - \Pr \left[\D\left(\left(X_i\right)_{i=1}^{q},\left(R_i\right)_{i=1}^{q}\right)=1 \right] \right| \leqslant \epsilon
\end{align*}
where the probability is over the choice of the random $X_i \leftarrow \{0,1\}^n$, the choice of a random key $K \leftarrow \{0,1\}^k$ and $R_i \leftarrow \{ 0,1\}^m$ conditioned on $R_i = R_j$ if $X_i = X_j$ for some $j < i$.
\end{definition}
Stream ciphers generate a keystream in a recursive manner. The security requires the output stream should be indistinguishable from uniform\footnote{We note that in a more standard notion the entire stream $X_1,\ldots,X_{q}$ is indistinguishable from random. This is implied by the notion above by a standard hybrid argument, with a loss of a multiplicative factor of $q$ in the distinguishing advantage.}.
\begin{definition}[Stream ciphers]
A \emph{stream-cipher} $\mathsf{SC} : \{0, 1\}^k \rightarrow \{0, 1\}^k \times \{0, 1\}^n$ is a function that need to be initialized with a secret state $S_0 \in \{0, 1\}^k$ and produces a sequence of output blocks $X_1, X_2, . . . $ computed as
\begin{align*}
 (S_i, X_i) := \mathsf{SC}(S_{i-1}).
\end{align*}
A stream cipher  $\mathsf{SC}$ is $(\epsilon,s,q)$-secure if for all $1 \leqslant i \leqslant q$, the random variable $X_i$ is $(s,\epsilon)$-pseudorandom given $X_1, . . . , X_{i-1}$ (the probability is also over the choice of the initial random key $S_0$).
\end{definition}
Now we define the security of leakage resilient stream ciphers, which follow the ``only computation leaks'' assumption.
\begin{definition}[Leakage-resilient stream ciphers]
A leakage-resilient stream-cipher is $(\epsilon,s,q,\lambda)$-secure if it is $(\epsilon,s,q)$-secure as defined above, but where the distinguisher in the $j$-th round gets $\lambda$ bits of arbitrary deceptively chosen leakage about the secret state accessed during this round. More precisely, before $(S_j,X_j) := \mathsf{SC}(S_{j-1})$ is computed, the distinguisher can choose any leakage function $f_j$ with range $\{0,1\}^{\lambda}$, and then not only get $X_j$, but also $\Lambda_j := f_j(\hat{S}_{j-1})$, where $\hat{S}_{j-1}$ denotes the part of the secret state that was modified (i.e., read and/or overwritten) in the computation $\mathsf{SC}(S_{j-1})$.
\end{definition}

\subsection{Security of leakage-resilient stream ciphers.} Best provable secure constructions of leakage-resilient stream ciphers are based on so called weak PRFs, primitives which look random when queried on random inputs  (\cite{Pietrzak2009,Faust2012,JetchevP14,DodisPietrzak2010,Yu2013}).
The most recent (TCC'14) analysis is based on a version of \Cref{thm:Simulator_Pietrzak}.
\begin{theorem}[Proving Security of Stream Ciphers \cite{JetchevP14}]\label{thm:streamcipher}
If $F$ is a $(\epsilon_F, s_F, 2)$-secure weak PRF then $\mathsf{SC}^F$ is a $(\epsilon', s', q, \lambda)$-secure leakage resilient stream cipher where
\begin{align*}
\epsilon' = 4q \sqrt{\epsilon_F2^{\lambda}},\quad  s' = \Theta(1) \cdot \frac{s_{F} \epsilon'^4}{2^{4\lambda}}.
\end{align*}
\end{theorem}
\begin{remark}[The exact complexity loss]\label{remark:exact_loss}
The inspection of the proof in \cite{JetchevP14} shows that $s_{F}$ equals the complexity of the simulator $h$ in \Cref{thm:Simulator_Pietrzak} applied to the class of all circuits of size $s'$, where $\epsilon$ is replaced by $\epsilon'$.
\end{remark}

\subsection{Time-Success Ratio}\label{sec:time-success}
The running time (circuit size) $s$ and success probability $\epsilon$ of attacks (practical and theoretical) against a particular primitive or protocol may vary. For this reason Luby~\cite{Luby1994} introduced the time-success ratio $\frac{t}{\epsilon}$ as a universal measure of security. This model is widely used to analyze provable security, cf. \cite{Buldas2013} and related works.
\begin{definition}[Security by Time-Success Ratio~\cite{Luby1994}]\label{def:time-success}
A primitive $P$ is said to be $2^{k}$-secure if for \emph{every adversary} with time resources (circuit size in the nonuniform model) $s$, the success probability in breaking $P$ (advantage) is at most $\epsilon < s\cdot 2^{-k}$. We also say that the time-success ratio of $P$ is $2^{k}$, or that is has $k$ bits of security.
\end{definition}
For example, $\mathsf{AES}$ with a $256$-bit random key is believed to have $256$ bits of security as a \emph{weak} PRF\footnote{We consider the security of $\mathsf{AES256}$ as a weak PRF, and not a standard PRF, because of non-uniform attacks which show that 
no PRF with a $k$-bit key can have $s/\epsilon \approx 2 ^k$ security~\cite{DeTT09}, at least unless we additionally require $\epsilon \gg 2^{-k/2}$.}.

\section{Proof of \Cref{thm:Simulator_ours}}\label{sec:proof:Simulator_ours}

For technical convenience, we attempt to efficiently approximate the conditional probability function $\Pr[Z=z|X=x]$ rather than building the sampler directly. Once we end with building an efficient approximation $h(x,z)$, we transform it into a sampler $h_{\mathsf{sim}}$ which outputs $z$ with probability $h(x,z)$  (this transformation yields only a loss of $2^{\ell}$).
We are going to prove the following fact
\begin{quote}
For every function $g$ on $\mathcal{X}\times\mathcal{Z}$ which is a $\mathcal{X}$-conditional probability mass function over $Z$ (that is
$g(x,z)\geqslant 0$ for all $x,z$ and $\sum_{z}g(x,z)=1$ for every $x$), and for every class $\cD$ closed under complements\footnote{This is a standard assumption in indistinguishability proofs. We can always extend the class by adding $-\D$ for every $\D\in\cD$, which increases the complexity only by 1.}
there exists $h$ such that
\begin{enumerate}[(a)]
\item $h$ is a $\mathcal{X}$-conditional probability mass function over $Z$
\item  $h$ is of complexity $s_h = O(2^{4\ell}\epsilon^{-2})$ with respect to $\cD$
\item $(X,Z)$ and $(X, h_{\mathsf{sim}}(X))$ are indistinguishable, which in terms of $g$ and $h$ means
\begin{align}
\left| \sum_{z} \E_{x\sim X} \left[  \D(x,z)\cdot ( g(x,z)-h(x,z) ) \right]\right|\leqslant \epsilon
\end{align} 
\end{enumerate}
\end{quote}
The sketch of the construction is shown in \Cref{alg:simulating}. Here we would like to point out two things. First,
we stress that we do not produce a strictly positive function; what our algorithm guarantees, is that the total negative mass is\emph{small}. We will see later that this is enough. Second, our algorithm performs essentially same operations for every $x$, which is why its complexity depends only on $\mathcal{Z}$.

For simplicity (and without losing generality) we assume $\mathcal{X}=\{0,1\}^n$ and $\mathcal{Z}=\{0,1\}^{\ell}$. We also denote for shortness $\overline{\D}(x,z)=\D(x,z)-\E_{z'\leftarrow U_{\mathcal{Z}}}\D(x,z')$ for any $\D$ (the "shift" transformation)

\IncMargin{1em}
\begin{algorithm}[ht!]
\SetKwInOut{Input}{input}\SetKwInOut{Output}{output}
\Input{Function $g: \nstrings\times\mstrings\rightarrow [0,1]$, accuracy paramter $\epsilon > 0$, class $\mathcal{D}$, step $\gamma$}
\Output{Function $h$ which is $\epsilon$-indistinguishable from $g$ under $\cD$, add up to 1 for every $x$,
and with total negative mass smaller $\gamma|\mathcal{Z}|^3$ }
$t \leftarrow 0$ \newline
$ h^{0}(x,z) \leftarrow \frac{1}{|\mathcal{Z}|} $ for every $x$ and $z$ \newline
\While(\tcc*[f]{while the simulator is not good enough}){ exists $\D\in\cD$ such that $\nllabel{condition} \E_{x\sim X} \left[ \sum_{z} \overline{\D}(x,z)\cdot \left(g(x,z') - h^{t}(x,z') \right) \right] \geqslant \epsilon $ }{
$\D^{t+1} \leftarrow \D$ \newline
\For(\tcc*[f]{improve the simulator towards the distinguisher direction}){$z'\in\mathcal{Z}$}{
$h^{t+1}(x,z')\leftarrow h^{t}(x,z')+\gamma\cdot\overline{\D^{t+1}}(x,z')  $} 
$t\leftarrow t+1$ \newline
$m\leftarrow 0$ \newline
\For( \tcc*[f]{locate the biggest negative point mass}){$z'\in\mathcal{Z}$}{
\If{${h}^{t}(x,z') < m$}{
$m\leftarrow \tilde{h}^{t}(x,z')$\newline
$z^{-} \leftarrow z'$
}
}
${h}^{t}(x,z^{-})=0$ \tcc*[f]{cut the biggest negative mass}
\nllabel{loop:fixing_positivity}\For{$z'\in\mathcal{Z}$}{
${h}^{t}(x,z') \leftarrow h^{t}(x,z')+\frac{m}{|\mathcal{Z}|-1}$\tcc*[f]{redestribute the cut mass}
}
}
\Return ${h}^{t}(x,z)$
\caption{Construct a Simulator}\label{alg:simulating}
\end{algorithm}
\DecMargin{1em}
\newpage
\begin{proof}
Consider the functions ${h}^{t}$. Define $\tilde{h}^{t+1}(x,z) \overset{def}{=} h^{t}(x,z)+\overline{\D}^{t+1}(x,z)$.
According to \Cref{alg:simulating}, 
we have
\begin{align}\label{eq:corrected_sim_recursive}
{h}^{t+1}(x,z) = {h}^{t}(x,z) + \gamma\cdot\overline{\D}^{t+1}(x,z) + \theta^{t+1}(x,z)
\end{align}
with the correction term $\theta^{t,r+1}(x,z)$ that be computed recursively as (see \Cref{loop:fixing_positivity} in \Cref{alg:simulating})
\begin{align}\label{eq:corrected_sim_correction_term}
\begin{array}{rl}
\theta^{t,0}(x,z) & = 0 \\
\theta^{t,r+1}(x,z) & = \left\{   
\begin{array}{rl}
-\min\left(h^{t}(x,z) + \gamma\cdot\overline{\D}^{t+1}(x,z),0\right), & \text{ if }  z= z_{\text{min}}^{t}(x)\\
 \frac{\min\left(h^{t}(x,z_{\text{min}}^{t}(x)))+\overline{\D}^{t+1}(x,z_{\text{min}}^{t}(x)),0\right)}{\#\mathcal{Z}-1} & \text{ if }z \not= z_{\text{min}}^{t}(x)
\end{array}
\right. \quad t=0,1,\ldots
\end{array}
\end{align}
where 
$z_{\text{min}}^{t}(x)$ is one of the points $z$ minimizing $ h^{t}(x,z)+\overline{\D}^{t+1}(x,z) $.
In particular
\begin{align}\label{eq:there_is_negative_mass}
 h^{t}(x,z_{\text{min}}^{t}(x)))+\overline{\D}^{t+1}(x,z_{\text{min}}^{t}(x)) < 0 \Longleftrightarrow \exists z: \ h^{t}(x,z)+\overline{\D}^{t+1}(x,z)<0
\end{align}
\underline{Notation}: for notational convenience we indenify the functions $\overline{D}^{t}(x,z)$, $\theta^{t}(x,z)$, $\tilde{h}^{t}(x,z)$ and $h^{t}(x,z)$ with matrices where $x$ are columns and $z$ are rows.
\begin{claim}[Complextity of \Cref{alg:simulating}]
$T$ executions of the ``while loop'' can be realized with time $O\left(T\cdot |\mathcal{Z}| \cdot \mathrm{size}(\cD)\right)$ and memory $O(|\mathcal{Z}|)$. 
\footnote{The RAM model}.
\end{claim}
This claim describes precisely resources required to compute the function $h^{T}$ for every $T$. In order to bound $T$, we define the energy function as follows:
\begin{claim}[Energy function]\label{claim:energy} Define the auxiliary function
\begin{align}\label{eq:energy}
\Delta^{t} = \sum_{i=0}^{t-1} \E_{x\sim X} \left[  \overline{\D}^{i+1}_x\cdot \left( g_x-h^{i}_x\right) \right].
\end{align}
Then we have $
\Delta^{t}  = E_1 + E_2
$
where
\begin{align}\label{eq:energy_split}
\begin{array}{rl}
E_1 &=\frac{1}{\gamma}\E_{x\sim X}\left[ \left( h^{t}_x-h^{0}_x\right)\cdot g_x  + \frac{1}{2}\sum_{i=0}^{t-1}\left(h^{i+1}_x  - h^{i}_x\right)^2-\frac{1}{2}\left( \left(h^{t}_x\right)^2-\left(h^{0}_x\right)^2 \right)\right] \\
E_2 & =  \frac{1}{\gamma}\E_{x\sim X}\left[ -\sum_{i=0}^{t-1} \theta^{i+1}_x\cdot \left( g_x-h^{i+1}_x\right) - \sum_{i=0}^{t-1} \theta^{i+1}_x\cdot \left( h^{i+1}_x-h^{i}_x\right) \right]
\end{array}
\end{align}
The proof is based on simple algebraic manipulations and appears in \Cref{proof:claim:energy}. 
\begin{remark}[Technical issues and intuitions]
From \Cref{eq:energy_split} it is clear that we need two important properties
\begin{enumerate}[(a)]
\item \emph{Boundedness of correction terms}, that is ideally $|\theta^{i}(x.z)| = O(\mathrm{poly}(|\mathcal{Z}|)\cdot\gamma)$.
\item \emph{Acute angle between the correction and the error}, that is $\theta^{i}_x\cdot (g_x-h^{i}_x) \geqslant 0$.
\end{enumerate}
Below we present an outline of the proof, discussing more technical parts in the appendix.
\subsubsection{Proof outline.}
Indeed, with these assumptions we can prove that 
\begin{align*}
E_1 + E_2 \leqslant O\left(\mathrm{poly}( |\mathcal{Z}|)\cdot \left( t \gamma + \gamma^{-1}\right) \right).
\end{align*}
Since in the other hand we have $t\epsilon\leqslant \Delta^{t} $, setting $\gamma = \epsilon/\mathrm{poly}(|\mathcal{Z})$ we get that the algorithm terminates after at most $T = \mathrm{poly}(|\mathcal{Z}|)\epsilon^{-2}$ steps. We stress that it outputs only a \emph{signed measure}, not a probability distribution. However, because of property (a) the negative mass is only of order $\mathrm{poly}(|\mathcal{Z}|)\epsilon$ and the function we end with can be simply rescaled (we replace negative masses by 0 and normalize the function dividing by a factor $1-m$ where $m$ is the total negative mass).
With this transformation, we replace the expected advantage $O(\epsilon)$ by slightly worse $O\left( \mathrm{poly}(|\mathcal{Z}|)\epsilon\right)$. We can then replace $\epsilon$ to get a clear dependency. Finally, we need to remember that we construct only a probability distribution function, not a sampler. Transforming it into a sampler yields an overhead of $O(\mathcal{Z})$. This discussion shows that it is possible to build a sampler of complexity $\mathrm{poly}(|\mathcal{Z}|)\epsilon^{-2}$. A more carefull inspection of the proof shows that we can actually achieve $|\mathcal{Z}|^5\epsilon^{-2}$.

\subsubsection{Technical Discussion}
We note that condition (b) somehow means that mass cuts should go in the right direction, as it is much simpler to prove that \Cref{alg:simulating} terminates when there are no correction terms $\theta^{t}$; thus we don't want to go in a wrong direction and ruin the energy gain. Concrete bounds on properties (a) and (b) are given in \Cref{claim:negative_mass_small,claim:increments_angle}.
\end{remark}
\end{claim}

In \Cref{alg:simulating} in every round we shift only one negative point mass (see \Cref{loop:fixing_positivity}). However, since this point mass is chosen to be as big as possible and since $h^{t+1}$ and $h^{t}$ differ only by a small term $\gamma\cdot\overline{\D}^{t+1}$ except the mass shift $\theta^{t+1}$, one can expect that we have the negative mass under control. Indeed, this is stated precisely in \Cref{claim:negative_mass_small} below.
\begin{claim}[The total negative mass is small]\label{claim:negative_mass_small}
Let 
\begin{align*}
\NM(h^{t}(x,\cdot)) = -\sum_{z} \min( h^{t}(x,z) ,0) 
\end{align*}
 be the total negative mass in $h^{t}(x,z)$ as the function of $z$. Then we have
\begin{align}\label{eq:negative_mass_bounded}
 \NM(h^{t}(x,\cdot) < |\mathcal{Z}|^3\gamma.
\end{align}
for every $x$ and every $t$.
\end{claim}
The proof is based on a recurrence relation that links $\NM(h^{t+1}(x,\cdot)$ with $\NM(h^{t}(x,\cdot)$, and appears in \Cref{proof:claim:negative_mass_small}.
\begin{claim}[The angle formed by the correction and the difference vector is acute] \label{claim:increments_angle}
For every $x,t$ we have $\textsf{Angle}\left(\theta^{t+1}_x,g_x-{h}^{t+1}_x\right)\in \left[-\frac{\pi}{2},\frac{\pi}{2}\right]$. 
\end{claim}
The proof appears in \Cref{proof:claim:increments_angle}.

\section{Time-success ratio under algebraic transformations}\label{sec:time-success_simpler}

In \Cref{thm:reduction_simpler} below we provide a quantitative analysis of how the time-success ratio changes under concrete formulas in security reductions.
\begin{lemma}[Time-success ratio for algebraic transformations]\label{thm:reduction_simpler}
Let $a,b,c$ and $A,B,C$ be positive constants. Suppose that $P'$ is secure against adversaries $(s',\epsilon')$, whenever $P$ is secure against adversaries $(s,\epsilon)$, where
\begin{align}\label{eq:reduction_simpler}
\begin{array}{rl}
 s' & = s\cdot c\epsilon^{C} - b\epsilon^{-B} \\ 
 \epsilon' & = a\epsilon^A.
\end{array}
\end{align}
In addition, suppose that the following condition is satisfied
\begin{align}\label{eq:small_exponent}
 A \leqslant C+1.
\end{align}
Then the following is true: if $P$ is $2^{k}$-secure, then $P'$ is $2^{k'}$-secure (in the sense of \Cref{def:time-success}) where
\begin{align}
k' = \left\{
\begin{array}{rl}
\frac{A}{B+C+1} k + \frac{A}{B+C+1}(\log c - \log b)-\log a, & \quad b\geqslant 1 \\
\frac{A}{C+1} k + \frac{A}{C+1}\log c -\log a, & \quad b = 0
\end{array}
\right.
\end{align}
\end{lemma}
The proof is elementary though not immediate. It can be found in \cite{DBLP:journals/corr/Skorski15b}.
\begin{remark}[On the technical condition \eqref{eq:small_exponent}]
This condition is satisfied in almost all applications, at in the reduction proof typically $\epsilon'$ cannot be better (meaning higher exponent) than $\epsilon$. Thus, quite often we have $A\leqslant 1$.
\end{remark}

\end{proof}

\bibliographystyle{amsalpha}
\bibliography{citations}

\appendix

\section{More on the flaw in \cite{JetchevP14}}\label{sec:flaws}

In the original setting we have $\mathcal{Z} = \{0,1\}^{\lambda}$. In the proof of the claimed better bound $\bigO{s\cdot 2^{3\lambda}\epsilon^{-2}}$ there is a mistake on page 18 (eprint version), when the authors enforce a signed measure to be a probability measure by a mass shifting argument. The number $M$ defined there is in fact a function of $x$ and is hard to compute, whereas the original proof amuses that this is a constant independent of $x$. During iterations of the boosting loop, this number is used to modify distinguishers class step by step, which drastically blows up the complexity (exponentially in the number of steps, which is already polynomial in $\epsilon$). In the min-max based proof giving the bound $\bigO{s\cdot 2^{3\lambda}\epsilon^{-4}}$ a fixable flaw is a missing factor of $2^{\lambda}$ in the complexity (page 16 in the eprint version), which is because what is constructed in the proof is only a probability mass function, not yet a sampler~\cite{Pietrzak15_private}.

\section{Proof of \Cref{claim:energy}}\label{proof:claim:energy}

We can rewrite \Cref{eq:energy} as
\begin{align}\label{eq:energy2}
\Delta^{t} & = \frac{1}{\gamma}\E_{x\sim X}  \left[\sum_{i=0}^{t-1}  \left( \left( h^{i+1}_x-h^{i}_x\right)- \theta^{i+1}_x \right)\cdot \left( g_x-h^{i}_x\right) \right] \nonumber \\
& = \frac{1}{\gamma}\E_{x\sim X} \left[ \sum_{i=0}^{t-1}  \left( h^{i+1}_x-h^{i}_x\right)\cdot \left( g_x-h^{i}_x\right) -  \sum_{i=0}^{t-1} \theta^{i+1}_x\cdot \left( g_x-h^{i}_x\right) \right]
\end{align}
First, note that
\begin{align}\label{eq:energy_bound_1}
 \sum_{i=0}^{t-1}  \left( h^{i+1}_x-h^{i}_x\right)\cdot & \left( g_x-h^{i}_x\right) = \nonumber \\
 =&  \left( h^{t}_x-h^{0}_x\right)\cdot g_x - \sum_{i=0}^{t-1}h^{i}_x\cdot\left( h^{i+1}_x-h^{i}_x\right) \nonumber \\
=&  \left( h^{t}_x-h^{0}_x\right)\cdot g_x + \frac{1}{2} \sum_{i=0}^{t-1}\left(h^{i+1}_x-h^{i}_x\right)\cdot\left( h^{i+1}_x-h^{i}_x\right) + \nonumber \\
& \quad -  \frac{1}{2} \sum_{i=0}^{t-1}\left(h^{i+1}_x+h^{i}_x\right)\cdot\left( h^{i+1}_x-h^{i}_x\right) \nonumber \\
=&  \left( h^{t}_x-h^{0}_x\right)\cdot g_x  + \frac{1}{2}\sum_{i=0}^{t-1}\left(h^{i+1}_x  - h^{i}_x\right)^2-\frac{1}{2}\left( \left(h^{t}_x\right)^2-\left(h^{0}_x\right)^2 \right)
\end{align}
As to the second term in \Cref{eq:energy2}, we observe that
\begin{align}\label{eq:energy_bound_2}
-\sum_{i=0}^{t-1} \theta^{i+1}_x\cdot \left( g_x-h^{i}_x\right)  = -\sum_{i=0}^{t-1} \theta^{i+1}_x\cdot \left( g_x-h^{i+1}_x\right) - \sum_{i=0}^{t-1} \theta^{i+1}_x\cdot \left( h^{i+1}_x-h^{i}_x\right)
\end{align}

\section{Proof of \Cref{claim:negative_mass_small}}\label{proof:claim:negative_mass_small}

\begin{proof}[Proof of \Cref{claim:negative_mass_small}]
We start by comparing the total negative mass in the functions $h^{t+1}=h^{t}+\overline{\D}^{t+1}+\theta^{t+1}$ and ${h}^{t} $. Suppose first that $\tilde{h}^{t}(x,z_0) < 0$ where $z_0 =  z_{\text{min}}^{t}(x)$. Since $\sum_{z\not=z_0} \tilde{h}^{t+1}= 1-\tilde{h}^{t+1}(x,z_0)$, there exists $z_1$ such that $
 \tilde{h}^{t+1}(x,z_1) \geqslant \frac{1-\tilde{h}^{t+1}(x,z_0)}{|\mathcal{Z}| - 1} > 0
$.
Combining this with \Cref{eq:corrected_sim_correction_term} we obtain
\begin{align}\label{eq:positive_mass_big_decrease}
h^{t+1}(x,z_1) &=\tilde{h}^{t+1}(x,z_1) +\frac{\tilde{h}^{t+1}(x,z_0)}{|\mathcal{Z}| - 1} \geqslant \frac{1}{|\mathcal{Z}| - 1}
\end{align}
By \Cref{eq:most_negative_mass} we have
\begin{align}\label{eq:shifting_decreases_negative_mass}
 \sum_{z\in\mathcal{Z}} & \min\left( {h}^{t+1}(x,z),0 \right) =\sum_{z\in\mathcal{Z}}\min\left( \tilde{h}^{t+1}(x,z)+\theta^{t+1}(x,z),0 \right) \nonumber \\
& =\sum_{z\in\mathcal{Z}\setminus\{z_0,z_1\}} \min\left( \tilde{h}^{t+1}(x,z)+\frac{\tilde{h}^{t+1}(x,z_0)}{|\mathcal{Z}|-1 },0 \right) + \nonumber \\
& \quad + \min\left( \tilde{h}^{t+1}(x,z_1)+\frac{\tilde{h}^{t+1}(x,z_0)}{|\mathcal{Z}|-1},0\right) \nonumber \\
& \geqslant \sum_{z\leftarrow {\mathcal{Z}\setminus\{z_0,z_1\}}}\left( \min(\tilde{h}^{t+1}(x,z),0)+\frac{\tilde{h}^{t+1}(x,z_0)}{|\mathcal{Z}|-1 } \right) \nonumber \\
& \quad+ \min(\tilde{h}^{t+1}(x,z_0),0)-\tilde{h}^{t+1}(x,z_0) +\min\left( \tilde{h}^{t+1}(x,z_1),0\right) \nonumber \\
& = \sum_{z\in\mathcal{Z}}\min(\tilde{h}^{t+1}(x,z),0)-\frac{\tilde{h}^{t+1}(x,z_0)}{|\mathcal{Z}|-1} 
\end{align}
where the inequality line follows from $\tilde{h}^{t+1}(x,z_0) < 0$ and \Cref{eq:positive_mass_big_decrease}. But by the definition of $z_{\text{min}}^{t}(x)$ in \Cref{eq:most_negative_mass} we get
\begin{align}\label{eq:most_negative_mass_relative_to_total}
\tilde{h}^{t+1}(x,z_0) \leqslant \frac{1}{|\mathcal{Z}|-1}\cdot \sum_{z\in\mathcal{Z}}\min\left( \tilde{h}^{t+1}(x,z),0\right)
\end{align}
Combining \Cref{eq:shifting_decreases_negative_mass} and \Cref{eq:most_negative_mass_relative_to_total} we obtain
\begin{align}\label{eq:before_vs_after_shifting_1}
-\sum_{z\in\mathcal{Z}} & \min\left( {h}^{t+1}(x,z),0 \right) \leqslant -\left(1-\frac{1}{(|\mathcal{Z}|-1)^2}\right)\sum_{z\in\mathcal{Z}} \min\left( \tilde{h}^{t+1}(x,z),0 \right).
\end{align}
Since $| h^{t+1}(x,z)-\tilde{h}^{t}(x,z)| \leqslant \gamma$ by \Cref{eq:corrected_sim_recursive}, we get the following recursion
\begin{align}\label{eq:before_vs_after_shifting_2}
-\sum_{z\in\mathcal{Z}} \min\left( {h}^{t+1}(x,z),0 \right) \leqslant -\left(1-\frac{1}{(|\mathcal{Z}|-1)^2}\right)\sum_{z\in\mathcal{Z}} \min\left( {h}^{t}(x,z),0 \right) + |\mathcal{Z}|\gamma
\end{align}
which can be rewritten as
\begin{align}\label{eq:before_vs_after_shifting_2}
\NM\left( {h}^{t+1}(x,\cdot) \right) < \left(1-\frac{1}{| \mathcal{Z}|^2}\right)\NM\left(h^{t}(x,\cdot)\right) + |\mathcal{Z}|\gamma.
\end{align}
which is in addition trivially true if $\tilde{h}^{t+1}(x,z) \geqslant 0$ for all $z$. The result follows by expanding this recursion till $t=0$.
\end{proof}

\section{Proof of \Cref{claim:increments_angle}}\label{proof:claim:increments_angle}

\begin{proof}
If $\theta^{t+1}(x,z) = 0$ then there is nothing to prove. Suppose that $\theta^{t+1}(x,z) < 0$. 
Let $z_0=z_{\text{min}}^{t}(x)$. According to \Cref{eq:corrected_sim_correction_term} we have $\theta^{t+1}(x,z_0) = -\tilde{h}^{t+1}(x,z_0)$ and $\theta^{t+1}(x,z) = \frac{\tilde{h}^{t+1}(x,z_0)}{\#\mathcal{Z}-1}$ for $z\not=z_0$. Therefore
\begin{align}\label{eq:correction_vs_difference_1}
\theta^{t+1}_x\cdot\left(g_x-\tilde{h}^{t+1}_x\right)  & = -\tilde{h}^{t+1}(x,z_0)\left(g(x,z_0)-\tilde{h}^{t+1}(x,z_0)\right) +\nonumber \\
&\quad+ \sum_{z\not=z_0}\frac{\tilde{h}^{t+1}(x,z_0)}{|\mathcal{Z}|-1}\cdot\left(g(x,z)-\tilde{h}^{t+1}(x,z)\right) \nonumber \\
& = -\tilde{h}^{t+1}(x,z_0)\left(g(x,z_0)-\tilde{h}^{t+1}(x,z_0)\right)  \nonumber \\
&\quad -\frac{\tilde{h}^{t+1}(x,z_0)}{|\mathcal{Z}|-1}\left(g(x,z_0)-\tilde{h}^{t+1}(x,z_0)\right) 
\end{align}
and 
\begin{align}\label{eq:correction_vs_difference_2}
-\theta^{t+1}_x\cdot\theta^{t+1}_x = -\tilde{h}^{t+1}(x,z_0)\cdot \tilde{h}^{t+1}(x,z_0)\left(1 + \frac{1}{|\mathcal{Z}-1|}\right).
\end{align}
Putting \Cref{eq:correction_vs_difference_1,eq:correction_vs_difference_2} together we obtain
\begin{align*}
\theta^{t+1}_x\cdot\left(g_x-h^{t+1}_x\right) & = \theta^{t+1}_x\cdot\left(g_x-\tilde{h}^{t+1}_x\right) -\theta^{t+1}_x\cdot\theta^{t+1}_x  \\
& = -\left(1+\frac{1}{|\mathcal{Z}|-1}\right)\tilde{h}^{t+1}(x,z_0)\cdot g(x,z_0)
\end{align*}
which is positive because $\tilde{h}^{t,r}(x,z_0)<0$ and $g(x,z_0) \geqslant 0$. This proves \Cref{claim:increments_angle}. 
\end{proof}

\end{document}